%% file: main.tex
\documentclass[sigconf]{acmart}
\usepackage[english]{babel}
\usepackage[utf8]{inputenc}
\usepackage{csquotes}
\usepackage{hyperref}
\usepackage{breakurl}
\usepackage[capitalise,nameinlink]{cleveref}
\crefname{section}{Sect.}{Sect.}
\Crefname{section}{Section}{Sections}
\usepackage{xspace}
\usepackage{subcaption}
\usepackage{centernot}
\usepackage{color}

\usepackage[color=green!40]{todonotes}
\usepackage{dsfont}
\usepackage{subcaption}
\usepackage[binary-units]{siunitx}
\usepackage{upgreek}
\usepackage{xparse}
\usepackage{dsfont}
\usepackage{listings}
\usepackage{dsfont}
\usepackage{multirow}
\usepackage{booktabs}

\usepackage{enumitem}

\makeatletter
\renewcommand\@formatdoi[1]{\ignorespaces}
\makeatother
\acmISBN{}
\acmDOI{10.1145/1122445.1122456}

\usepackage[clock]{ifsym}

\usetikzlibrary{shapes,arrows,automata,backgrounds,patterns}

\include{commands}

\begin{document}

\title[FPGA Stream Monitoring]{FPGA Stream-Monitoring of Real-time Properties}
\author{Jan Baumeister}
\affiliation{%
	\institution{Saarland University}
	\department{Department of Computer Science}
	\city{Saarbrücken}
	\state{Saarland}
	\country{Germany}
}
\orcid{0000-0002-8891-7483}
\email{jbaumeister@react.uni-saarland.de}

\author{Bernd Finkbeiner}
\affiliation{%
  \institution{Saarland University}
  \department{Department of Computer Science}
  \city{Saarbrücken}
  \state{Saarland}
  \country{Germany}
}
\email{finkbeiner@react.uni-saarland.de}

\author{Maximilian Schwenger}
\affiliation{%
  \institution{Saarland University}
  \department{Department of Computer Science}
  \city{Saarbrücken}
  \state{Saarland}
  \country{Germany}
}
\orcid{0000-0002-2091-7575}
\email{schwenger@react.uni-saarland.de}

\author{Hazem Torfah}
\affiliation{%
  \institution{Saarland University}
  \department{Department of Computer Science}
  \city{Saarbrücken}
  \state{Saarland}
  \country{Germany}
}
\email{torfah@react.uni-saarland.de}


\acmConference[EMSOFT19]{International Conference on Embedded Software}{October 13 -- 18, 2019}{New York City}

\begin{abstract}
  An essential part of cyber-physical systems is the online evaluation of real-time data streams. Especially in systems that are intrinsically safety-critical, a dedicated monitoring component inspecting data streams to detect problems at runtime greatly increases the confidence in a safe execution. Such a monitor needs to be based on a specification language capable of expressing complex, high-level properties using only the accessible low-level signals. Moreover, tight constraints on computational resources exacerbate the requirements on the monitor. Thus, several existing approaches to monitoring are not applicable due to their dependence on an operating system. 

  We present an FPGA-based monitoring approach by compiling an \rtlola specification into synthesizable VHDL code. \rtlola is a stream-based specification language capable of expressing complex real-time properties while providing an upper bound on the execution time and memory requirements. The statically determined memory bound allows for a compilation to an FPGA with a fixed size. An advantage of FPGAs is a simple integration process in existing systems and superb executing time. 
  The compilation results in a highly parallel implementation thanks to the modular nature of \rtlola specifications. This further increases the maximal event rate the monitor can handle.

\end{abstract}

\maketitle

\renewcommand{\shortauthors}{Baumeister~et~al.}

\input{introduction}

\input{rtlola}
\input{compilation}
\input{casestudy}
\input{conclusion}
\begin{acks}
  This work was partially supported by the German Research Foundation (DFG) as part of the Collaborative Research Center Foundations of Perspicuous Software Systems (TRR 248, 389792660), and by the European Research Council (ERC) Grant OSARES (No. 683300).
\end{acks}
\bibliographystyle{plain}
\bibliography{bibliography}

\end{document}

%% file: commands.tex
\newcommand{\commentout}[1]{}

\DeclareFontFamily{U}{MnSymbolC}{}
\DeclareSymbolFont{MnSyC}{U}{MnSymbolC}{m}{n}
\DeclareFontShape{U}{MnSymbolC}{m}{n}{
    <-6>  MnSymbolC5
   <6-7>  MnSymbolC6
   <7-8>  MnSymbolC7
   <8-9>  MnSymbolC8
   <9-10> MnSymbolC9
  <10-12> MnSymbolC10
  <12->   MnSymbolC12%
}{}
\DeclareMathSymbol{\powerset}{\mathord}{MnSyC}{180}

\tikzstyle{state} = [circle, draw, text width=.5cm, text centered, minimum height=.5cm, minimum width=.5cm]
\tikzstyle{component} = [rectangle, draw, text width=8em, text centered, minimum height=3em, fill=white]
\tikzstyle{smcomponent} = [component, minimum width=3cm, minimum height=2.5cm]
\tikzstyle{transition} = [draw, -stealth']
\tikzstyle{signal} = [draw, -latex']
\tikzstyle{config} = [densely dotted]

\colorlet{eventcolor}{green!50!black}
\colorlet{periodiccolor}{blue!50!black}
\tikzstyle{event} = [draw=eventcolor, thin, fill opacity=.3, pattern=north west lines, pattern color=eventcolor]
\tikzstyle{periodic} = [draw=periodiccolor, thin, fill opacity=.3, pattern=north east lines, pattern color=periodiccolor]
\tikzstyle{signalname} = [text width=10em, minimum height=2em]
\tikzstyle{nameright} = [signalname, align=left]
\tikzstyle{namecenter} = [minimum height=2em, align=center, rotate=65]
\tikzstyle{nameleft} = [signalname, align=right]

\newcommand*{\eg}{e.g.\@\xspace}
\newcommand*{\ie}{i.e.\@\xspace}

\newcommand*{\wrt}{w.r.t.\@\xspace}
\makeatletter
\newcommand*{\etc}{%
    \@ifnextchar{.}%
        {etc}%
        {etc.\@\xspace}%
}
\newcommand*{\cf}{%
    \@ifnextchar{.}%
        {cf}%
        {cf.\@\xspace}%
}
\newcommand*{\etal}{%
    \@ifnextchar{.}%
        {et~al}%
        {et~al.\@\xspace}%
}
\makeatother

\NewDocumentCommand{\twopartdef}{ m m m o}{
  \left\{
    \begin{array}{ll}
      #1 & \mbox{if } #2 \\
      #3 & \IfNoValueTF{#4}{\text{otherwise}}{\mbox{if } #4}
    \end{array}
  \right.
}
\NewDocumentCommand{\threepartdef}{m m m m m o}{
  \left\{
    \begin{array}{lll}
      #1 & \mbox{if } #2 \\
      #3 & \mbox{if } #4 \\
      #5 & \IfNoValueTF{#6}{\text{otherwise}}{\mbox{if } #6}
    \end{array}
  \right.
}
\NewDocumentCommand{\fourpartdef}{m m m m m m m m o}{
  \left\{
    \begin{array}{llll}
      #1 & \mbox{if } #2 \\
      #3 & \mbox{if } #4 \\
      #5 & \mbox{if } #6 \\
      #7 & \IfNoValueTF{#8}{\text{otherwise}}{\mbox{if } #8}
    \end{array}
  \right.
}

\newcommand\handcraftedclock[4][2]{%
  \begin{tikzpicture}[scale=#1,line cap=round,line width=#1*3pt]
    \filldraw [fill=periodiccolor!20!white] (0,0) circle (2cm);
    \foreach \angle / \label in
      {0/3, 30/2, 60/1, 90/12, 120/11, 150/10, 180/9,
      210/8, 240/7, 270/6, 300/5, 330/4}
    {
      \draw[line width=#1*1pt] (\angle:1.8cm) -- (\angle:2cm);
      \draw (\angle:1.4cm) node[scale=#1]{\textsf{\label}};
    }
    \foreach \angle in {0,90,180,270}
      \draw[line width=#1*2pt] (\angle:1.6cm) -- (\angle:2cm);
      \node[draw=none,font=\tiny,text=black,scale=#1] at (0,.9cm) {EMSOFT 2019};
      \draw[rotate=90,line width=#1*2pt] (0,0) -- (-#2*30-#3*30/60:0.7cm); 
      \draw[rotate=90,line width=#1*1.5pt] (0,0) -- (-#3*6:1cm); 
      \draw[rotate=90,line width=#1*.6pt,red] (0,0) -- (-#4*6:1.2cm); 
      \path [fill=black] (0,0) circle (3pt);
      \path [fill=red] (0,0) circle (1.5pt);
  \end{tikzpicture}%
}

\newcommand{\clockicon}{\handcraftedclock[.085]{10}{8}{50}}

\newcommand\from\colon

\renewcommand{\epsilon}{\varepsilon}

\NewDocumentCommand{\bool}{O{}}{\mathds{B}^{#1}}
\newcommand{\reals}{\mathds{R}}
\newcommand{\naturals}{\mathds{N}}

\newcommand{\rtlola}{\textsc{RTLola}\xspace}
\newcommand{\lola}{\textsc{Lola}\xspace}

\newcommand*{\LLC}{\textsc{LLC}\xspace}
\newcommand*{\HLC}{\textsc{HLC}\xspace}

\newcommand*{\signal}[1]{\textsl{#1}}
\newcommand*{\register}[1]{\textbf{#1}}
\newcommand*{\component}[1]{\textsc{#1}}
\newcommand*{\smstate}[1]{\texttt{#1}}

\newcommand*{\numout}{{n^{\uparrow}}}
\newcommand*{\numwind}{{n^{w}}}
\newcommand*{\numin}{{n^{\downarrow}}}
\newcommand*{\numtrig}{{n^{*}}}
\newcommand*{\sizets}{{s_{\ts}}}
\newcommand*{\ts}{\mathit{ts}}
\newcommand*{\sizeev}{{s_{\mathit{ev}}}}

\newcommand*{\clkrate}{\mathit{\xi}}

\newcommand*{\sclk}{\signal{sclk}\xspace}
\newcommand*{\hclk}{\signal{hclk}\xspace}

\newcommand*{\numdl}{{\ensuremath{\#\mathit{dl}}}}

\newcommand*{\external}{\ensuremath{\mathit{external}}}

\newcommand*{\clkperiod}{\signal{\ensuremath{\xi}}}
\DeclareMathOperator{\csr}{csr}
\DeclareMathOperator{\dep}{dep}
\DeclareMathOperator{\fin}{fin}
\DeclareMathOperator{\map}{map}
\DeclareMathOperator{\layer}{layer}
\DeclareMathOperator{\evalexpr}{evalexpr}
\DeclareMathOperator{\tar}{tar}
\DeclareMathOperator{\capa}{\kappa}
\DeclareMathOperator{\wdep}{wdep}
\DeclareMathOperator{\dltarget}{dl\_target}
\DeclareMathOperator{\deadline}{dl}
\newcommand*{\hyperperiod}{\mathord{\Pi}}

\newcommand*\backlog{\ensuremath{\mathit{bl}}\xspace}
\newcommand*\buffer{\ensuremath{\mathcal{B}}\xspace}
\newcommand*\abuffer{\ensuremath{\widetilde{\buffer}}\xspace}
\DeclareMathOperator\extbuffer\oplus
\DeclareMathOperator{\shift}{<\kern-3pt<}
\DeclareMathOperator{\dld}{dld}
\newcommand*\buffsize{\ensuremath{\mathcal{L}}\xspace}
\DeclareMathOperator{\decbuffer}{dec}
\DeclareMathOperator{\size}{size}

\definecolor{CommentColor}{rgb}{0.16,0.60,0.16} 
\definecolor{eclipseBlue}{RGB}{42,0.0,255}
\definecolor{eclipseGreen}{RGB}{63,127,95}
\definecolor{eclipsePurple}{RGB}{127,0,85}

\lstdefinelanguage{Lola}{
  keywords=[0]{constant, input, output, trigger, timeinput, frequency},
  keywordstyle=[0]\color{eclipseBlue}\bfseries,
  keywords=[1]{int, bool, string, ipv4, ipv6, Float64},
  keywordstyle=[1]\color{eclipseGreen},
  keywords=[2]{if, then, else},
  keywordstyle=[2]\color{eclipseBlue},
  keywords=[3]{invoke, terminate, extend},
  keywordstyle=[3]\color{eclipsePurple},
  keywords=[4]{sqrt, ground_speed, abs},
  keywordstyle=[3]\color{eclipsePurple},
    sensitive=false,
    comment=[l]{//},
    morecomment=[s]{/*}{*/},
    morestring=[b]',
    morestring=[b]"
}

%% file: introduction.tex
\section{Introduction}\label{sec:intro}

With the growing autonomy of cyber-physical systems, the evaluation,
aggregation, and monitoring of real-time data have become essential for
ensuring the safety of the system.  A principled approach to building
such monitors is provided by stream-based specification languages like
\rtlola~\cite{rtlola,streamlab}. Input streams that collect data from sensors, networks, etc.,
are filtered and combined into output streams that contain data
aggregated from multiple sources and over multiple points in time such
as over sliding windows of some real-time length. Trigger
conditions over these output streams then identify critical
situations.

Previous work has been very successful in using stream-based
specifications for analyzing recorded data streams, such as the flight
data of drones~\cite{uav1,streamlab} and network traces~\cite{lola2}. However, tools that have been
developed for the offline analysis of recorded data cannot directly be
used for online monitoring, such as for an onboard monitoring
component on a drone. The reason is the substantial software overhead
of such offline tools.  Cyber-physical systems operate under
narrow constraints on the available resources. A monitor must,
specifically, process all data in real time and within the available
memory.

In this paper, we present a compilation approach that realizes \rtlola
specifications on field-programmable gate arrays (FPGAs).  FPGAs have
dramatic advantages over software-based solutions in terms of
processing speed due to the inherent parallelism, and also in terms of
other factors such as energy consumption, weight, and ease of
integration within the cyber-physical system.

In \rtlola, input streams are event-driven, \ie, without a
priori known frequencies; output streams are typically periodic. This
difference is reflected in the realization of the monitor as a
two-module architecture consisting of a high-level controller and a
low-level controller. The role of the high-level controller is to
receive the events, prepare stream evaluations and to schedule periodic
tasks.  The low-level controller then computes new stream values based on
the information received from the high-level controller and triggers
an alarm when appropriate.

A key challenge for the compilation is the treatment of sliding window
expressions. In general, there is no bound on the memory needed to the
store the potentially unbounded number of events received during the
time period of the window.  Our monitoring circuit splits the full
window into smaller chunks, where the data can be pre-aggregated
without loss of precision. As a result, the number of registers needed
for the monitor can (under some mild assumptions on the aggregation
functions) be determined statically.

The immediate compilation to a hardware description language allows us to achieve a high level of parallelism. 
For this, we analyze the specification to identify modular sub-structures and evaluate them in parallel. 
We showcase the impact of this analysis with a synthetic case study. 
Furthermore, we demonstrate the practicality of the compilation by presenting experimental data from two realistic case studies from avionics and network monitoring. 
Both case studies indicate that the compilation utilizes the benefits of hardware: the implementation is highly efficient, requires only a small board, and consumes less than \SI{2}{\watt} of power.

The main contribution of this paper is an automatic compilation of an \rtlola specification into an FPGA monitor. 
The resulting circuits have a clear structure following the formal \rtlola semantics. 
The monitor is decoupled from the observed system. 
Unlike instrumentation-based approaches~\cite{tessla2,tesslaold}, the monitor is independent of the origin of the data. 
Furthermore, there are no assumptions on the frequency of the inputs granted it is lower than the maximum clock frequency of the FPGA.

The monitor utilizes the inherently parallel nature of hardware: the high-level controller is organized into a pipeline architecture, which ensures that new events can enter the controller before the processing of the previous events has been completed. 
In the low-level controller, however, the evaluation order ensures that independent streams are processed in parallel. 
Moreover, the monitor is highly space and energy efficient. 
Unlike interpreter-based approaches~\cite{lola,tessla2}, which include a general-purpose runtime environment, the compiled circuit is strictly limited to the operations that actually occur in the specification. 
As a result, the monitors of our case studies are able to run on small FPGA boards with little power~($<$~\SI{2}{\watt}). 

\subsection{Related Work}\label{sec:relatedwork}
\input{related.tex}

%% file: related.tex
Most of the earlier work on formal runtime monitoring  was based on temporal logics \cite{Drusinsky:2000:TRA:645880.672089,Lee99runtimeassurance, Finkbeiner+Sipma/01/Checking,ltl,Kupferman:2001:MCS:569028.569032,Havelund:2002:SMS:646486.694486}. The approaches vary between inline methods that realize a formal specification as assertions added to the code to be monitored \cite{Havelund:2002:SMS:646486.694486}, or outline approaches that separate the implementation of the monitor from the one of the system under investigation \cite{Finkbeiner+Sipma/01/Checking}. Based on these approaches and with the rise of real-time temporal logics such as MTL \cite{MTL} and STL \cite{STL}, a series of works introduced monitoring algorithms for real-time properties \cite{RobustMonSTL,monitoringSTL,Basin:2015:MMF:2772377.2699444,aerial}.

First translations from temporal logics to monitoring circuits have been introduced with the tools FoCs~\cite{RTl2Circuit}, developed at IBM Haifa, P2V~\cite{P2V}, a compiler that translates assertions written in sPSL~\cite{sPSL} to Verilog code, BusMOP~\cite{busmop}, which synthesized monitor circuits from specifications written in past-time linear temporal logic for monitoring PCI bus traffic, and MBAC~\cite{Boule:2008:AAS:1297666.1297670}, an automata-based monitor synthesizer for PSL properties. 
Inspired by these constructions, an optimized approach for bounded future properties was presented in~\cite{ltl2circuits}. Hardware runtime monitors for real-time properties were presented by Jaksic~et~al.~\cite{stl2fpga}\@, where  monitors for  STL specifications  were implemented in an FPGA. Further work on FPGA implementations of real-time temporal specification was introduced with the tool R2U2~\cite{r2u2tool,r2u2}, an outline monitoring approach that allows for monitoring specifications in MTL including future-time specifications.

Temporal logics come with the advantage of providing formal guarantees on the space and time complexity of the synthesized monitors. However, a major drawback of these logics is their expressiveness. 
When monitoring cyber-physical systems, one needs to express properties beyond yes and no verdicts (for example with some degree of arithmetic operation) to be able to monitor realistic properties of the system. 
Stream-based languages over complex datatypes like \rtlola~\cite{rtlola,streamlab} provide such expressiveness and further maintain a desirable level of formal guarantees.

The stream-based approach to monitoring was pioneered by the specification language \lola~\cite{lola}. \lola\ is related to  
synchronous programming languages like Lustre~\cite{lustre,lustre2}, and Esterel~\cite{esterel}, which have been widely used for the development of digital circuits~\cite{surveysyncproglang}. In contrast to these languages, \lola\ is a descriptive language, which subsumes the temporal logics and can express both past and future properties.
A feature of \lola\ is that upper bounds on the memory required for monitoring can be computed statically. \rtlola\ extends \lola\ with asynchronous streams and real-time features such as sliding windows.
Two other extensions of \lola\ are TeSSLa and Striver.
TeSSLa~\cite{tessla2} allows for monitoring piece-wise constant signals where streams can emit events at different speeds with arbitrary latencies. It relies on the instrumentation of C code and is thus not independent of the monitored system. Moreover, \rtlola\ comes with the feature of computing aggregations over sliding windows, and allows for the decoupling of the computation of output streams from variable input event rates via fixed-rate clocks. The main difference between \rtlola\ and Striver~\cite{striver} is that \rtlola\ has both variable-rate and fixed-rate streams and provides convenient, native operators such as sample-and-hold and sliding windows that translate between the two types of streams. The fixed rate in \rtlola\ allows for a more direct translation to a hardware implementation of the monitor.   

An approach for compiling synchronous \lola\ has been presented in~\cite{maltry}. We remove the assumption of synchronously arriving data and add real-time capabilities to the specification language.

%% file: rtlola.tex
\section{RTLola}\label{sec:rtlola}
\rtlola\cite{rtlola} is a \emph{stream-based} specification language with real-time features based on the specification language \lola~\cite{lola}. In stream-based runtime monitoring, sensor readings are interpreted as streams of input data. This streams are fed into a stream engine that computes new sequences of data called output streams based on the values of input streams. The output streams compute statistics over the sensor data and allow for stating verdicts about the monitored system. The computation rules for output streams are defined in \rtlola by a stream equation, which is a defining equation that maps a stream variable to a stream expression. 
Consider for example a GPS module in a drone that delivers data about the current longitude and latitude, and a monitor that checks if the GPS module is delivering data in appropriate frequencies. An \rtlola specification for defining such a monitor is given by the following stream definitions:
\begin{lstlisting}
input gps: (Float64, Float64)
output gps_glitch: Bool@1Hz:= 
  gps.aggregate(over:2s,using:count) < 10
trigger gps_glitch "GPS sensor frequency < 5Hz"
\end{lstlisting}

The stream \lstinline{gps} is an input stream that represents the readings of the GPS module and is expected to deliver data with a frequency greater than or equal to \SI{5}{\hertz}. 
To check whether this data is delivered with the expected frequency, we define the output stream \lstinline{gps_glitch} that computes a sliding window with a duration of two seconds over the stream \lstinline{gps}. 
The stream \lstinline{gps_glitch} is computed in a frequency of \SI{1}{\hertz} and checks whether ten values are received from the GPS module in the last two seconds. 
The window over the input stream \lstinline{gps} is computed via the expression \lstinline{gps.aggregate(over:2s,using:count)}, which counts the number of data values of \lstinline{gps} in the last two seconds. 
If the number of values is less than 10, then \lstinline{gps_glitch} evaluates to true. 
In this case, an alarm is raised with the message \lstinline{"GPS sensor frequency < 5Hz"}. This alarm is defined by the trigger expression \lstinline{trigger gps_glitch}. 

The stream above is a \emph{periodic} stream and as such computed at a fixed frequency.
In addition to that, \rtlola also allows for the definition of \emph{event-based} streams by omitting the frequency. 
Event-based streams are evaluated whenever streams occurring in its stream expression are evaluated.
For example, if we want to check whether a vehicle is slowing down, we can compute the change in velocity between the last two velocity sensor readings:
\begin{lstlisting}
input velo: Float64	
output slowing_down: Bool := 
  velo - velo.offset(by:-1).defaults(to:0) < 0
\end{lstlisting}
The stream \lstinline{slowing_down} is computed every time \lstinline{velo} receives a new value. To compute the difference, the stream expression uses the \emph{offset operator} to access the last (\lstinline{.offset(by:-1)}) and current value of the stream \lstinline{velo} and then compute the difference between these two values. In case the value of an offset operation is not defined, the default operator (\lstinline{.defaults(to:d)}) returns the value $d$. In the example above, \lstinline{velo.offset(by:-1)} is not defined before receiving the first velocity reading, so the default value $0$ is used instead.

In the case where an output stream is defined over more than one stream, the output stream is evaluated only if all streams it depends on are evaluated as well. If one of these values is missing, one can still enforce the computation of the stream using the \emph{sample-and-hold operator} (\lstinline{.hold()}). This operator accesses the last value computed for a stream. If it is not present, the provided default operator (\lstinline{.defaults(to:d)}) is used. The following specification clarifies the role of this operator.
\begin{lstlisting}
input gps: (Float64, Float64)
input height: Float64
output too_low: Bool := if zone(gps) 
  then (height.hold().defaults(to:300)) < 300 
  else false
trigger too_low "Flying low in inhabited area"
\end{lstlisting}
The function \lstinline{zone} determines whether the drone is in an inhabited area. When the vehicle is in this area, the specification checks whether its current height (\lstinline{height}) is less than 300 feet. If this is the case, an alarm is raised because it violates the flight regulations for inhabited areas.

\rtlola imposes some rules on how streams may access the values of other streams. 
\Cref{fig:rtlolaframework} shows the general picture of \rtlola specifications. 
\begin{figure}[t]
  \scalebox{1}{
    \input{figures/streamaccesses}
  }
  \caption{Stream accesses of event-based and periodic streams in \rtlola}
  \label{fig:rtlolaframework}	
\end{figure}
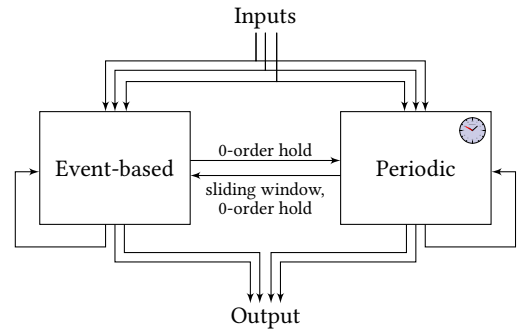
The values of an output stream may be used in the definitions of other output streams as long as the following rules are respected:

\noindent\textit{Access via sliding window:} Periodic streams may access values of other streams via a sliding window without any further restriction. 

\noindent\textit{Access via offset operator:}
When accessing a stream with the offset operator, an \rtlola specification must respect the following rules:

\noindent\textit{1. Accessing periodic streams in event-based streams:}
These accesses are only allowed with the sample-and-hold operator. 

\noindent\textit{2. Accessing event-based streams in event-based streams:}
These accesses are always valid. However, the accessing stream is only extended if all accessed streams are extended at the same time. The sample and hold operation eliminates this dependency.

\noindent\textit{3. Accessing event-based streams in periodic streams:} 
Periodic streams only access event-based streams with the sample-and-hold operator. 

\noindent\textit{4. Accessing periodic streams in periodic streams:} 
A periodic stream $s$ may access the values of another periodic stream $s'$ if and only if the frequency of $s'$ is an integer-multiple of the one of $s$. Otherwise the access is only allowed via the sample-and-hold operator. 

\noindent\textit{5. Recursive stream access:} 
Any stream is allowed to access its own history of values as long as it does not create any circular access like accessing itself with an offset of 0. Consider the following specification:
\begin{lstlisting}
output num_glitches: UInt32 := 
  num_glitches.offset(by:-1).defaults(to:0) + 
    (if gps_glitch then 1 else 0)
\end{lstlisting}
  The output stream is an event-based stream that is evaluated every time a new value is computed for \lstinline{gps_glitch}. Note that there is no need for the sample-and-hold operator as the output stream only depend on \lstinline{gps_glitch}. If the new value of \lstinline{gps_glitch} is true, then the new value of \lstinline{num_glitches} is computed by increasing its last value (\lstinline{num_glitches.offset(by:-1).defaults(to:0)}) by one. Otherwise, if \lstinline{gps_glitch} is false, the new value of \lstinline{num_glitches} is equal to its last one. 

For the full syntax and type system of \rtlola\ we refer the reader to the technical report\footnote{The technical report also describes parametrization with dynamic stream creation, which we do not consider here.}~\cite{rtlola}.

In the rest of the paper we use the variables $\numout$, $\numin$ and $\numtrig$ to indicate the number of output streams, number of input streams and number of triggers in an \rtlola specification, respectively.

\subsection{Monitoring RTLola Specifications}\label{sec:monitoring}

Monitoring an \rtlola specification consists of receiving events, evaluating stream expressions, and triggering alerts when necessary. The separation of event-based and periodic streams manifests itself in the monitoring algorithm in that it consists of an event-based and a periodic process. 

The event-based process receives an event and extends streams according to the \emph{evaluation order} $\prec$, \ie, if the stream expression of stream $s$ contains a lookup with target $s'$, then $s' \prec s$. Thus, $s'$ needs to be extended before $s$. The event-based process respects this by successively evaluating streams as soon as the evaluation order permits it.

The periodic process schedules streams according to their frequency. Since all frequencies are determined a priori, we can compute an array of \emph{deadlines}, where deadline $D_i$ is a delay $d_i$ and a set of streams $S_i$ such that when $D_{i-1}$ was due, after $d_i$ seconds, $S_i$ need to be evaluated. The least common multiple of the periods of all periodic streams is the \emph{hyper-period}~$(\hyperperiod{})$ and $\numdl$ denotes the number of deadlines within one hyper-period. Like the event-based process, the periodic process also respects the evaluation order.

An \rtlola specification can be monitored in one of two modes. Offline mode describes a monitoring process that happens after the fact based on log data. It is useful for post-mortem analyses or for validating a specification based on previous system runs. Online mode, however, is the concurrent execution of a system and its monitor. FPGA-based monitoring is especially interesting for the online mode because this mode requires timely processing of events and imposes tighter restrictions on the monitor in terms of available resources. 

The major difference between the two modes in the evaluation process is the source of the current timestamp. In online mode, the value is the system time of the monitor. In offline mode, however, events are annotated with time stamps. The monitor considers the received time stamp to be the current time and checks whether a deadline would have been missed. If so, it first computes all periodic streams affected by the deadline. Afterwards, it processes the event as described before.

\subsection{Sliding Windows}
The evaluation of sliding windows needs special attention. Assume the stream expression of $s$ with frequency \si{x\hertz} contains a sliding window expression such as \lstinline{s$'$.aggregate(over:$
\delta$,using:$\gamma$)} for some duration $\delta$ and aggregation function $\gamma$. 
A naive implementation requires to store all values of $s'$ within the last $\si{\delta\second}$, which is unfeasible because there is no information about the arrival frequency of $s'$. %
If $\gamma\from A^\ast \to B$ is a \emph{list~homomorphism} as defined by Maarten~\cite{pub:2686}, the sliding window can be evaluated accurately with only a finite amount of memory. 
List homomorphisms can be split into four components: a unary $\map \from A \to T$ and finalization $\fin \from T \to B$, an associative binary reduction $\oplus\from T \times T \to T$, and a neutral element $\epsilon$ \wrt $\oplus$.
Assuming $\gamma$ is a list homomorphism, we utilize the fact that sliding windows only occur in periodic streams. 
All new values occurring within a \si{x\second} time interval are effectively equivalent \wrt their arrival time. 
We now apply the \emph{bucketing} approach proposed by Li~\etal~\cite{DBLP:journals/sigmod/LiMTPT05} and split the duration of the window into $\delta x^{-1}$ equal-sized buckets. 
Each bucket stores an intermediate value, initialized with $\epsilon$, and pre-aggregates all values within two evaluations of the window expression using $\oplus$. 
At the time of the evaluation, the intermediate values get reduced to obtain the final value.
  
Fortunately, many commonly used aggregation functions are list homomorphisms, such as summation, minimization, maximization, counting, integration, and averaging.

As an example, consider the following specification:
\begin{lstlisting}
input velo : Float32 
output avg_velo @1Hz := 
  velo.aggregate(over:3s,using:avg)
      .defaults(to:8.0)
\end{lstlisting}
Since the average is a list homomorphism, we define the following concrete components:
\begin{itemize}[itemindent=-10pt]
  \item $\map\from\reals\to\reals\times\naturals$ with $\map(v) \coloneq (v,1)$
  \item $\fin\from\reals\times\naturals\to\reals$ with $\fin(v,c) \coloneq \frac{v}{c}$
  \item $\oplus\from(\reals\times\naturals)^2\to\reals\times\naturals$ with $(v_1,c_1) \oplus (v_2, c_2) \coloneq (v_1 + v_2, c_1 + c_2)$
  \item $\epsilon \coloneq (0,0)$
\end{itemize}

\Cref{fig:tab:avg} details the computation of the average with three buckets. We list the values for all buckets at points in time when either an event arrives or \lstinline{avg_velo} gets computed. Here, $p_1$ represents the ``oldest'' bucket, and $p_3$ the most recent one.

\begin{figure}[t]
\begin{tabular}{r r r r r r r}
  \toprule
	Event & Time & \texttt{velo} & $p_1$ & $p_2$ & $p_3$ & \texttt{avg\_velo} \\
  \midrule
	   & \SI{0.0}{\second} &      & $\epsilon$ & $\epsilon$ & $\epsilon$ & \\
	 1 & \SI{0.5}{\second} & 10.0 & $\epsilon$ & $\epsilon$ & (10.0,1)   & \\
	 2 & \SI{0.6}{\second} & 10.1 & $\epsilon$ & $\epsilon$ & (20.1,2)   & \\
	   & \SI{1.0}{\second} &      & $\epsilon$ & $\epsilon$ & (20.1,2)   & 8.0\\
	   & \SI{2.0}{\second} &      & $\epsilon$ & (20.1,2)   & $\epsilon$ & 8.0\\
	 3 & \SI{2.2}{\second} & 9.9  & (20.1,2)   & $\epsilon$ & (9.9,1)    & \\
	   & \SI{3.0}{\second} &      & (20.1,2)   & $\epsilon$ & (9.9,1)    & 10.0\\
  \bottomrule
\end{tabular}
\caption{Detailed computation of a sliding average.}
\label{fig:tab:avg}
\end{figure}
Initially, all buckets contain the element $\epsilon$.
Upon receiving the first velocity at time stamp \SI{0.5}{\second}, the value of the last bucket is changed to $(0,0) \mathbin{\oplus}\map(10.0) = (10.0,1)$.
When the next event is received at time stamp \SI{0.6}{\second}, we add the value to the same bucket and get $(10.0,1) \mathbin{\oplus}\map(10.1) = (20.1,1)$.
At time stamp \SI{1.0}{\second}, we compute \lstinline{avg_velo} for the first time. Since the current time stamp is less than the length of the window, the default values is used.
Afterwards, we evict the oldest bucket, shift all bucket values to the left, and add a new one with value $\epsilon$.
The same happens at time stamp \SI{2.0}{\second}.
The next event arrives at time stamp \SI{2.2}{\second} and is added to $b_3$. At time \SI{3}{\second}, we stop using the default value and aggregate the buckets. The resulting value is finalized, \ie, $\fin((20.1,2) \mathbin{\oplus} (0,0) \mathbin{\oplus} (9.9,1)) = \frac{30}{3} = 10$

%% file: figures/streamaccesses.tex
\begin{tikzpicture}
  \node[draw, minimum height=1.5cm, minimum width=2cm](event)    at (-2,-2) {Event-based};
  \node[draw, minimum height=1.5cm, minimum width=2cm](periodic) at (2,-2)  {Periodic};
  \node (inputs) at (0,0) {Inputs};
  \foreach \offset in {0,.15,-.125} {
    \draw[signal] (\offset,-.2) -- ++(0,-.5-\offset) -- ++(-2,0) -- ++(0,-.55+\offset);
    \draw[signal] (\offset,-.2) -- ++(0,-.5-\offset) -- ++(+2-2*\offset,0) -- ++(0,-.55+\offset);
  }
  
  \newcommand\InOutLower{-2.75}
  
  \node (outputs) at (0,-4) {Output};
  \foreach \offset in {0,-.125} {
    \draw[signal] (-2-\offset,\InOutLower) -- ++(0,-.5-\offset) -- ++(+1.8,0) -- ++(0,-.55+\offset);
    \draw[signal] (+2+\offset,\InOutLower) -- ++(0,-.5-\offset) -- ++(-1.8,0) -- ++(0,-.55+\offset);
  }
  
  \draw[signal] (-2.125,\InOutLower) -- ++(0,-.3) -- ++(-1.2,0) -- ++(0,+1) -- ++(+.32,0);
  \draw[signal] (+2.125,\InOutLower) -- ++(0,-.3) -- ++(+1.2,0) -- ++(0,+1) -- ++(-.32,0);
  
  \draw[signal] (-1,-1.9) -- ++(+2,0);
  \node[align=center] () at (0,-1.75) {\footnotesize 0-order hold};
  \draw[signal] (+1,-2.1) -- ++(-2,0);
  \node[align=center] () at (0,-2.4) {\footnotesize sliding window,\\[-4pt] \footnotesize 0-order hold};
  
  \node () at (2.75,-1.5) {\clockicon};

%
%
	
\end{tikzpicture}

%% file: compilation.tex
\section{Compilation}\label{sec:compilation}
The hardware realization of an \rtlola specification consists of two modules connected via a first-in-first-out queue as can be seen in \Cref{fig:overalldesign}. 
The \emph{High-level Controller (\HLC)} receives external events consisting of event data for each affected input stream and a time stamp in offline mode, as well as the system time in online mode. 
The \HLC acts as mediator between event-based inputs and periodic deadlines, such that later components in the architecture do not need to distinguish them anymore.
The number of bits the \HLC receives is $\sizets + \sum_{i=1}^\numin(s_i + 1)$ where $\sizets$ and $s_i$ denote the number of bits required to represent a single timestamp and value of input stream $i$, respectively. 
The additional bit per input stream indicates whether the current event contains a new value for the respective stream. 
The \HLC decides whether a periodic deadline is due or an event ought to be evaluated. This decision is based on information about events and the internal system clock.
The respective information is preprocessed with respect to the specification and stored in the \emph{Queue}.
It consists of $\sizeev = (\sum_{i=1}^{\numin} (s_i + 1)) + \sizets + \numout$ bits with the following semantics:
\begin{enumerate}
  \item $\sum_{i=1}^\numin (s_i + 1)$ bits encode an event as explained before. If the signal encodes a deadline, all bits are 0 indicating that no data is available. 
  \item $\sizets$ bits contain the time stamp used for the evaluation of sliding windows and as implicitly defined input stream with name \lstinline{time}.
  \item $\numout$ bits declare for each output stream whether or not they are affected by the current deadline or event. 
\end{enumerate}
The \emph{Low-level Controller (\LLC)} uses this information to manage the evaluation process: all input streams, and output streams which expression can be evaluated immediately are extended first, followed by the remaining output streams in further steps according to the evaluation order. The \LLC also manages updates and the evaluation of sliding windows occurring in output stream expressions.

Due to the lower complexity of \HLC's task, it is capable of receiving events faster than the \LLC can process them. For this reason, the queue acts as a buffer between the two components. While this does not prevent a loss of data when the pressure on the evaluator exceeds its limits for an extended amount of time, it temporarily relieves the stress of a sudden burst of events. Moreover, it cleanly decouples the two components, enabling them to work independently and concurrently at their own pace.
\input{figures/overall}

\subsection{Notation}

We first introduce some notation. The $\circ$ operator denotes bit-con\-ca\-te\-na\-tion. 
$0^n$ denotes an $n$-fold concatenation of $0$-bits. Let $x$ be a bit string of length $n$. 
$x[i]$ denotes the $i$th bit of $x$ assuming $i < n$. 
$x[\ell\dots u]$ is the substring $x[\ell] \circ x[\ell+1] \circ \dots \circ x[u-1]$ for $\ell < u < n$. 
The bounds can be omitted, \ie, $x[\dots u] = x[0\dots u]$ and $x[\ell\dots] = x[\ell\dots n]$. 
Further, let $\clkrate$ be the internal system clock rate and sums over all input streams are abbreviated by omitting the limits, \ie, $\sum s_i = \sum_{1 \leq i}^\numin s_i$.

We distinguish between signals and registers. The former are data lines between components, which we will write in a slanted font, such as $\signal{signal}$. The latter are mere flip-flop components that are updated with a rising clock edge, written in bold face: $\register{register}$.

\subsection{High-level Controller}

\input{figures/hlc}

This module receives external events and schedules periodic tasks. It pre-processes data with respect to the specification and stores the information in the queue.

\Cref{fig:schematic:hlc} shows the schematic of the module. Dotted lines represent signals and components that are only present in the offline mode. 
The \HLC has access to the common system clock \sclk, and two registers \register{avail} and \register{din} which are written by an external entity and contain data of new events.
The components are organized in a pipeline architecture, which ensures that new events can enter the controller before the processing of the previous events has been completed.
The green, top-left-striped part handles the event-based inputs, whereas the blue, top-right striped part handles periodic deadlines. The \component{HLQInterface} then unifies events and deadlines.

\paragraph*{\component{PreScaler}} 
This component scales the system clock \sclk down by a constant factor to the \HLC-internal \hclk clock. \hclk drives the \component{Scheduler}, \component{EventDelay}, and the \component{ExtInterface}. The \component{PreScaler} also provides an internal clock for the \component{HLQInterface}, which ticks twice as fast as \hclk and slower than \sclk. For a cleaner illustration, \Cref{fig:schematic:hlc} does not include the respective data lines, as well as \signal{valid} bits accompanying every data line with width greater than 1 indicating the presence of meaningful data on the wire. 

\paragraph*{\component{ExtInterface}} 
This component handles the communication with external input sources. 
The external source writes a 1-bit latch \register{avail} when new input data is available in the \register{din} register. 
In online mode, the \component{ExtInterface} reads \register{din}, and forwards it to the \component{EventDelay}. 
In offline mode, the input event also contains a time stamp, which the \component{ExtInterface} extracts and forwards to the \component{TimeSelect} component. 
In both modes, it then  clears \register{avail}, indicating that the next event can be received.

Formally, \component{ExtInterface} waits on \hclk and behaves as follows, where \signal{ev} carries the event data, \signal{ext\_ts} is the external time stamp received with the event, and \signal{valid\_ext\_ts} indicates whether there is new and valid data on the \signal{ext\_ts} wire.
\begin{align*}
  \signal{ev}^0 &= 0^{\sum s_i} \\
  \signal{ev}^{t+1} &= \twopartdef{\register{din}^t[\sizets\dots]}{\register{avail}^t}{0^{\sum s_i}} \\
  \register{avail}^0 &= 0 \\
  \register{avail}^{t+1} &= \twopartdef{1}{\external^t \land \neg \register{avail}^t}{0} \\
  \signal{ext\_ts}^0 &= 0^\sizets \\
  \signal{ext\_ts}^{t+1} &= \twopartdef{\register{din}^t[\dots\sizets]}{\register{avail}^t}{0^\sizets} \\
  \signal{valid\_ext\_ts}^0 &= \signal{valid\_ev}^0 = 0 \\
  \signal{valid\_ext\_ts}^{t+1} &= \signal{valid\_ev}^{t+1} = \register{avail}^t
\end{align*}
Here, $\external$ is an oracle indicating a change depending on an external event.

\paragraph*{\component{TimeSelect}} 
The component waits on the system clock and computes the internal time stamp $\signal{its}$. 
In offline mode, this is simply the time stamp formerly extracted from the input event. 
Thus, this component boils down to a simple wire and does not introduce any delay in the signal. 
In online mode, however, this component computes the time that has passed so far by repeatedly adding the period $\clkperiod$ of the system clock.
This component uses an internal register \register{reg\_its} mirroring the value of \signal{its}. It persists the value of the signal without introducing a delay\footnote{This can be achieved by letting the input wire of the register carry the same signal as the output wire.}.
\begin{align*}
  \register{reg\_its}^0 &= 0^{\sizets} \\
  \register{reg\_its}^{t+1} &= \register{reg\_its}^t + \clkperiod = (t+1) * \clkperiod \\
  \signal{its}^t &= \register{reg\_its}^t \\
  \signal{valid\_its}^t &= 1
\end{align*}

\paragraph*{\component{Scheduler}}
This component inspects the current internal timestamp \signal{its} and detects when a periodic stream is due. 
It first determines the start time and stores it in the \register{period} register: in online mode that is simply $0^\sizets$, whereas in offline mode this is the first time stamp received from the external source. 
It then maintains the invariant that \register{period} contains the least time stamp in the current hyper-period. 
If, for example, the specification contains two periodic streams with frequency $\SI{2}{\hertz}$ and $\SI{5}{\hertz}$, then the hyper-period is $\SI{1}{\second}$. 
If the first received event carries the timestamp $\SI{3.4}{\second}$, \register{period} remains $\SI{3.4}{\second}$ until a time stamp greater than or equal to $\SI{3.3}{\second} + \hyperperiod{} = \SI{4.4}{\second}$ is received. 
In this case, it jumps to $\SI{4.4}{\second}$.
As a result, the difference between \signal{its} and \register{period} represents the time within the current hyper-period.

The register \register{did} contains the id of the current deadline, \ie, the deadline that needs to be evaluated next, in unary encoding.
The encoding is a trade-off: a binary encoding requires fewer registers and wires but also two decoders, one in the \component{Scheduler} and one in the \component{HLQInterface}. 
The \register{did} register is initialized with $0^\numdl$, which is an invalid unary number and indicates that the \component{Scheduler} has not been initialized, \ie, it did not receive a start time, yet. 
The initialization takes place in the first cycle in online mode, or in the first cycle with enabled \signal{valid\_ext\_ts} bit in offline mode.

Lastly, the \signal{prog}(ress) signal indicates whether a new deadline is due. 
It checks whether the \component{Scheduler} was initialized and whether the position in the current hyper-period exceeds the current deadline. 
For this check, it accesses the statically determined array of deadline offsets as described in \Cref{sec:monitoring}. 
The lookup consists of conjoining each element of the array with the respective bit of the \register{did} and bitwise disjoining all results:
$\deadline(\register{did}) = \bigvee_{i=1}^\numdl \deadline_i \land \register{did}[i]$

In the following definitions, a subscript $\mathit{off}$ ($\mathit{on}$) indicates the offline (online) version of the register or signal. Usages without subscript use the respective version.
\begin{align*}
  \signal{init}_{\mathit{off}^0} &= 0 \\
  \signal{init}_{\mathit{off}}^{t+1} &= \signal{valid\_its}^{t+1} \land (\register{did}^{t} = 0^{\numdl}) \\
  \signal{init}_{\mathit{on}}^{t} &= \twopartdef{1}{t = 1}{0} \\
  \register{did}^0 &= 0^\numdl \\
  \register{did}^{t+1} &= \threepartdef%
      {10^{\numdl-1}}%
      {\signal{init}^{t+1}}%
      {\csr(\register{did}^{t})}%
      {\neg\signal{init}^{t+1} \land \signal{prog}^{t+1}}%
      {\register{did}^{t}}
\end{align*}
\begin{align*}
  \register{period}^0 &= 0^\numdl \\
  \register{period}_{\mathit{off}}^{t+1} &= \threepartdef%
      {\signal{its}^{t+1}}%
      {\signal{init}^{t+1}}%
      {\register{period}_{\mathit{off}}^{t} + \hyperperiod}%
      {\register{did}^{t} = 0^{\numdl + 1}1 \land \signal{prog}^{t+1}}%
      {\register{period}_{\mathit{off}}^{t}} \\
  \register{period}_{\mathit{on}}^{t+1} &= \threepartdef%
      {0}%
      {\signal{init}^{t+1}}%
      {\register{period}_{\mathit{on}}^{t} + \hyperperiod}%
      {\register{did}^{t} = 0^{\numdl + 1}1 \land \signal{prog}^{t+1}}%
      {\register{period}_{\mathit{on}}^{t}}\\%
  \signal{prog}^{t+1} &= \register{did}^{t} \neq 0^\numdl \land (\signal{its}^{t+1} - \signal{period}^{t}) > \deadline(\register{did}^{t}) \\ 
\end{align*}

Here, $\csr$ is a 1-bit cyclic shift to the right. The output signals are thus defined as:\\
\begin{minipage}[t]{.45\columnwidth}
\begin{align*}
  \signal{hold}_{\mathit{on}}^t &= 0 \\
  \signal{dl}^t &= \signal{its}^t \circ \register{did}^t
\end{align*}
\end{minipage}
\begin{minipage}[t]{.45\columnwidth}
\begin{align*}
  \signal{hold}_{\mathit{off}}^t &= \signal{prog}^t \\
  \signal{valid\_dl}^t &= \lnot \signal{prog}^t
\end{align*}
\end{minipage}

\paragraph*{\component{EventDelay}} 
This component composes the internal time stamp and the current event. The time stamp is later used in the evaluation process. In online mode, the compound signal is then passed to the \component{HLQInterface} without delaying the signal.

In offline mode, however, the \component{EventDelay} needs to take the \signal{hold} signal into account. 
To compensate for the delay introduced by the \component{Scheduler}, the compound signal is delayed by one cycle. 
Afterwards, the data is delayed further until \signal{hold} turns off. 
During the hold period, new events can be received and need to be stalled. We discuss this issue below. 

Formally, the component waits on $\hclk$ and uses two internal registers, \register{data} which introduces the mandatory one-cycle delay and \register{reg\_tev} mirroring the signal \signal{tev}.
\begin{align*}
  \register{data}^0 &= 0^{1+\sizets+\sum (s_i + 1)} \\
  \register{data}^{t+1} &= \twopartdef{\register{data}^t}{\signal{hold}^{t+1}}{\signal{valid\_ev}^{t+1} \circ \signal{its}^{t+1} \circ \signal{ev}^{t+1}} \\
  \register{stalled}^0 &= 0^{1+\sizets+\sum (s_i + 1)} \\
  \register{stalled}^{t+1} &= \twopartdef{\register{stalled}^t}{\signal{hold}^{t+1}}{\register{data}^t} \\
  \signal{tev}^t &= \signal{stalled}^t[1\dots] \\
  \signal{valid\_tev}^{t+1} &= \neg\signal{hold}^{t+1} \land \signal{tev}^t[0]
\end{align*}
Note that \signal{ev} and \signal{its} are always valid at the same point in time, so we can verify the invariant 
\[ \forall t\colon \signal{valid\_ev}^t \iff \signal{valid\_its}^t \]

\paragraph*{QInterface}
This component accepts data from the \component{EventDelay} and the \component{Scheduler} and forwards information to the queue. 
It can only push one data packet per cycle to the queue. 
Both in offline and online mode, however, it can receive a deadline and an event at the same time. 
For this reason, this component is clocked twice as fast as \hclk. This enables it to wait on events in even cycles and wait on deadlines in odd cycles. Yet, it needs to be slower than \sclk such that the queue can still process both data packets in time.
As a result, it grants precedence to events. This is desired to compensate for the delay introduced by the \component{EventDelay} and preserve the correct order of events and deadlines.

Formally, in even cycles this component computes:
\begin{align*}
  \signal{push}^t &= \signal{valid\_ev}^t \\
  \signal{data}^t &= \signal{ev}^t \circ \bigvee_{i=1}^\numin (\dep(i) \land \signal{ev}^t[\sum_{j=1}^i (s_j + 1) - 1])
\end{align*}
Here, $\dep$ is another static array of $\numout$ bit wide registers where each bit represents a dependency between streams. I.e., if $\dep(i)[j]$ is on, output stream $j$ transitively depends on input stream $i$ and thus has to be evaluated with the current event. The respective dependencies are conjoined with $\signal{ev}[\sum_{j=1}^i (s_j + 1) - 1]$, \ie, the bit indicating whether the current event carries a new value for input stream $i$. Overall, the data sent to the queue thus contains the event data, the time stamp of the event, and one bit per stream indicating whether the stream will be evaluated.

In odd cycles, the $\signal{data}$ signal only contains the streams affected by the deadline:
\begin{align*}
  \signal{push}^t &= \signal{valid\_dl}^t \\
  \signal{data}^t &= 0^{\sum (s_i + 1)} \circ \signal{dl}^t[\dots \sizets] \circ \dltarget(\signal{dl}^t[\sizets\dots])
\end{align*}

\subsection{Input Buffering}
The stalling mechanism in the \component{EventDelay} and \component{Scheduler} is only necessary in offline mode. 
Two consecutive events $e_i$ and $e_{i+1}$ can have time stamps that skip several deadlines. 
In this case, the \component{Scheduler} repeatedly considers $e_{i+1}$ as a new value and triggers the computation of a deadline until no more deadline is due. 
During this time, it raises the \signal{hold} flag, so that the \component{EventDelay} stalls $e_{i+1}$ before sending it to the \component{HLQInterface}. 
While stalling, the \component{ExtInterface} can continue receiving events that are either lost, or override $e_{i+1}$. 
To prevent this, we add an input buffer of size \buffsize in front of the \component{Scheduler} and \component{EventDelay}. 
The required buffer size can be computed based on the input data. 
Assume that the \HLC receives a new input value every $\delta$ \hclk cycles.
The \emph{backlog} $\backlog(e_i)$ describes how many cycles it takes to fully process all entries currently in the buffer when receiving event $e_i$, including all deadlines induced by $e_i$.
\begin{align*}
  \backlog(e_1) &= 0 \\
  \backlog(e_{i+1}) &= \backlog(e_i) - \min\{\backlog(e_{i}), \delta - 1\} + \dld(e_{i+1})
\end{align*}
Here, $\dld(e_i)$ is the number of periodic deadlines that become due when receiving $e_i$. 
Intuitively, between event $e_i$ and $e_{i+1}$, $\delta-1$ cycles pass without a new event, so we either process $\delta-1$ deadlines or events, or all entries in the buffer. 
Upon receiving $e_{i+1}$, we need to process an additional $\dld(e_i)$ deadlines plus the new event. 
At the same time, another cycle passes, so we can immediately process one event or deadline. 
This effectively eliminates the incoming event, so only $\dld(e_i)$ needs to be taken into account.
  
Let \buffer be a buffer of size \buffsize with the following semantics, where $\buffer^\eta_i$ is the $i$th entry of $\buffer$ at cycle $\eta$:
\begin{align*}
  \buffer^0 &= \{ \bot \}^\buffsize \\
  \buffer^{\eta+1} &= \fourpartdef{%
      \buffer^\eta \shift 1
    }{%
      \neg \signal{hold}^{\eta + 1} \land \neg \signal{valid\_its}^{\eta+1}
    }{%
      \buffer^\eta
    }{%
      \hphantom{\neg} \signal{hold}^{\eta + 1} \land \neg \signal{valid\_its}^{\eta+1}
    }{%
      \buffer^\eta \oplus \signal{its}^{\eta + 1}
    }{%
      \hphantom{\neg} \signal{hold}^{\eta + 1} \land \hphantom{\neg} \signal{valid\_its}^{\eta+1}
    }{%
      (\buffer^\eta \shift 1) \extbuffer \signal{its}^{\eta + 1}
    }{%
      \neg \signal{hold}^{\eta + 1} \land \hphantom{\neg} \signal{valid\_its}^{\eta+1}
    }
\end{align*}
Here, $\buffer\shift 1$ shifts the entire buffer content to the left, \ie, the first and thus oldest entry gets evicted, the $n+1$st entry becomes the $n$th, and the last entry becomes $\bot$. 
$\buffer\extbuffer\nu$ denotes that the first free entry of \buffer, \ie, the first $k$ with $\buffer_k = \bot$, is replaced by $\nu$. 
If no such entry exists, the buffer overflows. 
Formally, the theorem states the following:
\begin{theorem}\label{thm:buffsize}
  If the buffer size $\buffsize$ maximizes $\backlog$, the buffer will never overflow: 
  \[
    \buffsize \geq \max \{\backlog\} \implies \forall \eta\colon \neg \signal{valid\_its}^\eta \lor \neg \signal{hold}^\eta \lor \buffer^\eta_\buffsize = \bot
  \]
\end{theorem}

\input{proof}

\subsection{Low-level Controller (\LLC)}\label{sec:llc}
\input{figures/llc}
\input{figures/evalcontrol}
\input{figures/streams}

The \LLC receives elements from the queue and evaluates streams according to the information received. After the evaluation, it checks for violated properties and triggers an alarm if appropriate. 

As can be seen in \Cref{fig:schematic:llc}, it consists of a \component{LLQInterface} component which communicates with the queue and triggers an evaluation process taking place in the \component{EvalController}.

\paragraph*{LLQInterface}
This component consists of a three-state machine depicted in \Cref{fig:statemachine:llqif}. 
In the \smstate{idle} state, it waits on new inputs from the queue. 
On a falling edge of \signal{empty}, it transitions into the \smstate{pop} state, rising the \signal{pop} signal for one \sclk cycle. 
At the end of this cycle, it unconditionally transitions to \smstate{eval}, setting the \emph{evaluation enable} \signal{een} latch. 
This signals the \component{EvalController} that valid data is on the \signal{d\textsubscript{in}} wire, so an evaluation can be started.
After the evaluation is completed, \component{EvalController} clears the \signal{een} signal. 
Depending on the current queue state, it transitions back to \smstate{idle} or \smstate{pop}.

\paragraph*{EvalController}

This component is a state machine as depicted in \Cref{fig:statemachine:ec} with $\ell+2$ states where $\ell = \max(\ell \in \mathds{N} | \exists s_1 \dots s_\ell\colon s_1 \prec \dots \prec s_\ell)$ is the number of layers of the evaluation order (see \Cref{fig:statemachine:ec}). 
In addition to the state machine, there are $\numin/\numout/\numwind$ input/output/window components. 
In the following, components and signals indexed with $i, j, \eta$ refer to inputs, outputs, and windows, respectively.

In the \smstate{idle} state, the \component{EvalController} waits on a rising edge of $\signal{een}$, on which it transitions to state \smstate{1}. 
This state corresponds to a so-called \emph{pseudo-extension} phase, where all output streams that get a new value in this evaluation cycle are extended by a pseudo value~$\#$. 
This value will never be used in a computation but allows for resolving offsets correctly without shifting the offsets depending on the evaluation status of the target stream. 
Input streams are immediately extended by their new values, and windows evict outdated buckets. Thus:
\begin{align*}
  &\forall i \leq \numin\colon  \signal{upd}_i = \signal{d\textsubscript{in}}[\sum_{n\leq i} (s_n + 1)] \\
  &\forall j \leq \numout\colon \signal{upd}_j = \signal{d\textsubscript{in}}[\sum(s_i + 1) + \sizets + j] \\
  &\forall \eta \leq \numwind\colon \signal{evict}_\eta = 1
\end{align*}
The structure of input, output and window components is depicted in $\Cref{fig:schematic:innercomps}$. In the input stream components we get the following behavior for a rising edge in \signal{upd\textsubscript{i}} where $\capa(i)$ describes the greatest offset of any lookup with target $i$:
\begin{align*}
  \signal{done}^t &= \signal{upd}^t \\
  \register{R}^0_n &= 0^{s_i + 1} \\
  \register{R}^{t+1}_n &= \threepartdef{%
            \register{R}^t_{n+1}%
          }{%
            \signal{upd}^{t+1} \land n \neq \capa(s_i)
          }{%
            \register{R}^t_n
          }{%
            \neg\signal{upd}^{t+1}
          }{%
            \signal{d\textsubscript{in}}^{t+1}\circ 1
          }[%
            \signal{upd}^{t+1} \land n = \capa(s_i)
          ] \\
  \signal{d\textsubscript{out}}^0 &= 0^{\capa(i)\cdot (s_i + 1)} \\
  \signal{d\textsubscript{out}}^{t+1} &= \register{R}^t_1 \circ \dots \circ \register{R}^t_{\capa(s_i)}
\end{align*}
By storing $\capa(i)$ values for any stream $i$, all offsets can be resolved when evaluating stream expressions.

Output streams on a rising edge of \signal{pe} behave as follows:
\begin{align*}
  \signal{done}^t &= \signal{pe}^t \\
  \register{R}_n^0 &= 0^{\capa(j)\cdot(s_j+1)} \\
  \register{R}_n^{t+1} &= \twopartdef{\#}{n = \capa(j)}{\register{R}_{n+1}^t} \\
  \signal{d\textsubscript{out}}^0 &= 0^{\capa(j)\cdot(s_j + 1)} \\
  \signal{d\textsubscript{out}}^{t+1} &= \register{R}^t_1 \circ \dots \circ \register{R}^t_{\capa(s_j)}
\end{align*}
For windows, the number of buckets is $\beta$, \ie, the length of the window $\mathit{dur}_\eta$ multiplied with the extend frequency $f_\eta$ of stream in which the window occurs. 
On a rising edge of \signal{evict}, \signal{d\textsubscript{in}} carries the current time stamp in the first $\sizets$ bits. 
The window requires this information to decide whether new buckets are outdated. 
If so, the values of all registers are shifted and the now-empty bucket is initialized with $\epsilon$. 
The internal \register{T} register stores the time when the next bucket becomes outdated.
\begin{align*}
  \register{T}^0 &= 0^\sizets \\
  \register{T}^{t+1} &= \twopartdef{%
            \register{T}^t
          }{%
            \signal{d\textsubscript{in}}[\dots\sizets] \leq \register{T}^t
          }{%
            \register{T}^t + f_\eta
          } \\
  \signal{done}^0 &= 0 \\
  \signal{done}^{t+1} &= \signal{d\textsubscript{in}}[\dots\sizets] \leq \register{T}^t\\
  \register{R}_n^0 &= \epsilon \\
  \register{R}_n^{t+1} &= \threepartdef{
            \epsilon
          }{
            n = \beta \land \signal{d\textsubscript{in}}[\dots\sizets] > \register{T}^t
          }{
            \register{R}_{n+1}^t
          }{
            n \neq \beta \land \signal{d\textsubscript{in}}[\dots\sizets] > \register{T}^t
          }{
            \register{R}_n^t
          }[
            \signal{d\textsubscript{in}}[\dots\sizets] \leq \register{T}^t
          ]
\end{align*}
Signal \signal{done\textsubscript{1}} indicates that phase 1 of the evaluation is complete:
\[
  \signal{done\textsubscript{1}} = \bigwedge_{i \leq \numin} \signal{upd}_i \implies \signal{done}_i \land \bigwedge_{j \leq \numout} \signal{pe}_j \implies \signal{done}_j \land \bigwedge_{\eta \leq \numtrig} \signal{done}_\eta
\]
Note that the implication ensures that a \signal{done} signal is only relevant, if the respective component was enabled.

After \signal{done\textsubscript{1}} is raised, the \component{EvalController} transitions to phase 2 via state \smstate{2.1}. 
In the \smstate{2.x} states, streams are successively extended according to the evaluation order and windows are updated whenever the target stream computed a new value. 
Wires connect streams and windows \wrt their dependencies, \ie, all streams output a sequence of values coupled with a bit indicating its validity. 
Invalid values are then replaced with the default values specified in the stream expression. 
Window lookups require an additional computation step, initiated by the \signal{req} signal.

Formally, when transitioning to state \smstate{2.x} with $1 \leq x \leq \ell$, the \component{EvalController} raises the update signals for outputs and windows if appropriate,\ie, if the stream is in the respective evaluation layer and the \HLC indicated that the stream is affected.
\begin{align*}
  \signal{eval}_j &= j \in \layer(x) \land \signal{d\textsubscript{in}}[\sum (s_i + 1) + \sizets + j] \\
  \signal{upd}_\eta &= \signal{d\textsubscript{out\textsubscript{$\tar(\eta)$}}}[s_{\tar(\eta)}]
\end{align*}
On a rising edge of $\signal{eval}_j$, the output stream computes its new value and updates its internal state:
\begin{align*}
  \signal{done}^t &= \signal{eval}^t \\
  \register{R}_n^{t+1} &= \twopartdef{\evalexpr(j) \circ 1}{\signal{eval}^{t+1} \land n = \capa(j)}{\register{R}_{n}^t} \\
  \signal{d\textsubscript{out}}^{t+1} &= \register{R}^t_1 \circ \dots \circ \register{R}^t_{\capa(s_j)}
\end{align*}
Here, $\evalexpr(j)$ is the result of evaluating the stream expression of stream $j$. The computation can be split into several computation steps depending of the size of the expression to increase the maximum system clock frequency. In this case, the \signal{done} bit cannot be set immediately after receiving the \signal{eval} command.
Note that only $\#$ values are overwritten in this step and the valid bit is set. 
In sliding windows, a new values is added by applying the $\mathit{map}$ function and reducing it onto the last bucket.
\begin{align*}
  \register{R}^{t+1}_\beta &= \register{R}_\beta^t \oplus \map(\signal{d\textsubscript{in}}^{t+1}) \\
  \signal{done}^t &= \signal{upd}^t
\end{align*}
It requires an additional step to compute the new value of the sliding window. 
This process is initiated by the \component{EvalController} by raising the \signal{req\textsubscript{$\eta$}} flag after the window's target stream was computed. All bucket values get reduced using the aggregation's reduction function $\oplus$, and finalized afterwards. 
Since $\oplus$ is associative and the number of buckets is a compile time constant, the reduction is structured as a binary tree with logarithmic depth in the number of buckets.
This triggers the following behavior in the window:
\begin{align*}
  \signal{d\textsubscript{out}}^{t+1} &= \fin(\register{R}_1^{t} \oplus \dots \oplus \register{R}_\beta^{t}) \\
  \signal{done}_{2.x}^t &= \signal{rq}^t 
\end{align*}

%% file: figures/overall.tex
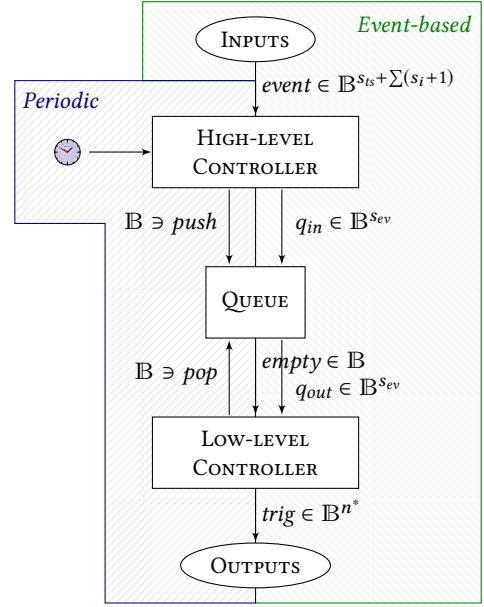
\begin{figure}[t]
\centering
\hspace*{-.8cm}
\begin{tikzpicture}
  \node [draw, ellipse, fill=white] at (0,1.5) (inputs) {\component{Inputs}};
  \node [component] at (0,0) (hlc) {\component{High-level Controller}};
  \node at (-2.5,0) (time) {\clockicon};
  \node [component, text width = 3.5em] at (0,-2) (q) {\component{Queue}};
  \node [component] at (0,-4) (llc) {\component{Low-level Controller}};
  \node [draw, ellipse, fill=white] at (0,-5.5) (outputs) {\component{Outputs}};

  \path [signal] (inputs) -- (hlc);
  \node [nameright] at (1.65,.95) {$\signal{event} \in \bool[\sizets + \sum (s_i + 1)]$};
  \path [signal] (time) -- (hlc);
  
  \path [signal] (-.35,-0.5) -- (-.35,-1.5); 
  \node [nameleft] at (-2.1,-0.95) {\hfill$\bool \ni \signal{push}$};
  \path [signal] (.35,-0.5) -- (.35,-1.5); 
  \node [nameright] at (2.1,-0.95) {$\signal{q\textsubscript{in}} \in \bool[\sizeev]$};

  \path [signal] (-.35,-3.5) -- (-.35,-2.5); 
  \node [nameleft] at (-2.1,-2.95) {\hfill$\bool \ni \signal{pop}$};
  \path [signal] (0,-2.5) -- (0,-3.5); 
  \node [nameright] at (1.65,-2.8) {$\signal{empty} \in \bool$};
  \path [signal] (.35,-2.5) -- (.35,-3.5); 
  \node [nameright] at (2.1,-3.15) {$\signal{q\textsubscript{out}} \in \bool[\sizeev]$};

  \path [signal] (llc) -- (outputs);
  \node [nameright] at (1.65,-4.8) {$\signal{trig} \in \bool[\numtrig]$};
  
  \begin{scope}[on background layer]    
    \path[event] (-1.5,2) -- (-1.5,.95) -- (0,.95) -- (0,-6) -- (3,-6) -- (3,2) -- cycle;
    \node () at (2.1, 1.7) {\sl\textcolor{eventcolor}{Event-based}};
    \path[periodic] (-3.2,.95) -- (0,.95) -- (0,-6) -- (-2,-6) -- (-2,-.95) -- (-3.2,-.95) -- cycle;
    \node () at (-2.6, .7) {\sl\textcolor{periodiccolor}{Periodic}};
  \end{scope}
  
\end{tikzpicture}

\caption{Schematic of an \rtlola monitor composed of two modules connected via a queue. The \emph{High-level Controller} manages the order in which periodic and event-based streams have to be evaluated. The \emph{Low-level Controller} manages the evaluation process of all affected streams.}
\label{fig:overalldesign}
\end{figure}

%% file: figures/hlc.tex
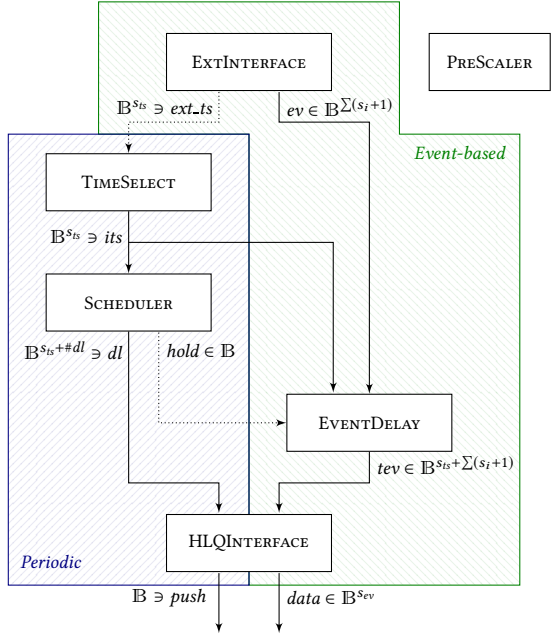
\begin{figure}[t]
  \centering
  \hspace*{-1cm}
  \begin{tikzpicture}[scale=.8, every node/.style={transform shape}]
  
    \node [component, text width=1.8cm] at (6,0)  (prescaler) {\component{PreScaler}};
    \node [component] at (2,0)  (extinter) {\component{ExtInterface}};
    \node [component] at (0,-2) (timeselect) {\component{TimeSelect}};
    \node [component] at (0,-4) (scheduler) {\component{Scheduler}};
    \node [component] at (4,-6) (eventdelay) {\component{EventDelay}};
    \node [component] at (2, -8) (qinter) {\component{HLQInterface}};

    \path [signal, config] (1.5,-.5) -- (1.5,-1) -| (timeselect); 
    \node [nameleft] at (-.2,-.75) {\hfill$\bool[\sizets] \ni \signal{ext\_ts}$};
    \path [signal] (2.5,-.5) -- (2.5,-1) -| (eventdelay); 
    \node [nameright] at (4.2,-.75) {$\signal{ev} \in \bool[\sum (s_i + 1)]$\hfill};
    \path [signal] (0,-3) -| (3.4,-5.5); 
    \path [signal] (timeselect) -- (scheduler);
    \node [nameleft] at (-1.7,-2.9) {\hfill$\bool[\sizets] \ni \signal{its}$};
    \path [signal] (scheduler) |- (1.5,-7) -- (1.5,-7.5);
    \node [nameleft] at (-1.7,-4.8) {$\hfill\bool[\sizets + \numdl] \ni \signal{dl}$};
    \path [signal, config] (.5,-4.5) |- (eventdelay);
    \node [nameright] at (2.2,-4.8) {$\signal{hold} \in \bool$\hfill};
    \path [signal] (eventdelay) |- (2.5,-7) -- (2.5,-7.5);
    \node [nameright] at (5.7,-6.7) {$\signal{tev} \in \bool[\sizets + \sum (s_i + 1)]$\hfill};
    \path [signal] (1.5,-8.5) -- (1.5,-9.5);
    \node [nameleft] at (-.3,-8.9) {\hfill$\bool \ni \signal{push}$};
    \path [signal] (2.5,-8.5) -- (2.5,-9.5);
    \node [nameright] at (4.2,-8.9) {$\signal{data} \in \bool[\sizeev]$\hfill};
    
    \begin{scope}[on background layer]    
      \path[event] (-.5,1) -- (-.5,-1.2) -- (2,-1.2) -- ++(0,-7.5) -- ++(4.5, 0) -- ++(0,7.5) -- ++(-2, 0) -- ++(0, 2.2) -- cycle;
      \node () at (5.5, -1.5) {\sl\textcolor{eventcolor}{Event-based}};
      \path[periodic] (-2,-1.2) --        (2,-1.2) -- ++(0,-7.5) -- ++(-4,0) -- cycle;
      \node () at (-1.3, -8.3) {\sl\textcolor{periodiccolor}{Periodic}};
    \end{scope}

  \end{tikzpicture}
  \caption{Schematic of the High-level Controller receiving external events, managing periodic deadlines, and preparing data for the Low-level Controller.}
  \label{fig:schematic:hlc}
\end{figure}

%% file: proof.tex
\begin{proof}
  We define an abstract buffer \abuffer where each abstract entry corresponds to a concrete one in \buffer. 
  Its value states how many clock cycles are required to process the deadlines induces by the respective concrete entry if it were the first one.
  \begin{align*}
    \abuffer^0 &= \{ \bot \}^\buffsize \\
    \abuffer^{\eta+1} &= \fourpartdef{
        \decbuffer(\abuffer^{\eta})
      }{
        \abuffer^{\eta}_1 > 0 \land \neg \signal{valid\_its}^{\eta +1}
      }{
        \decbuffer(\abuffer^\eta) \oplus \dld(\signal{its}^{n+1})
      }{
        \abuffer^{\eta}_1 > 0 \land \hphantom{\neg} \signal{valid\_its}^{\eta +1}
      }{
        \abuffer^\eta <<1}{\abuffer^{\eta}_1 = 0 \land \neg \signal{valid\_its}^{\eta +1}
      }{
        (\abuffer \shift 1) \extbuffer \dld(\signal{its}^{\eta+1})
      }{
        \abuffer^{\eta}_1 = 0 \land \hphantom{\neg} \signal{valid\_its}^{\eta +1}
      }
  \end{align*}
  Here, $\decbuffer(\abuffer)$ reduces the value of the first and thus oldest value by one, which represents that a deadline induced by the event was processed.
  We define the size of an entry in $\abuffer$ as 
  \[
    \size(\abuffer^\eta_i) = \twopartdef{0}{\abuffer^\eta_i = \bot}{\abuffer^\eta_i + 1}
  \]

The proof follows from three facts. 

\emph{1) $\backlog$ is the sum of the size of $\abuffer$'s entries}, \ie, for any event $e_i$ that reaches the buffer in cycle $\eta_i$, the following holds:
\begin{equation}\label{lem:backlogabuff}
  \backlog(e_i) = \sum_{j=1}^\buffsize \size(\abuffer^{\eta_i})
\end{equation}

Proof by induction on the event sequence consisting of the events $e_1,e_2,\dots$. Assume $\eta_i$ is the clock cycle in which $e_i$ arrives at the buffer. For $\eta_0$: 
\[
  \sum_{j=1}^\buffsize \size(\abuffer^{\eta_0}_0) = \sum_{j=1}^\buffsize 0 = 0 = \backlog(e_0)
\]
In the induction step, we go from $\eta_i$ to $\eta_{i+1}$. 
Note that these two points in time are separated by $\delta$ clock cycles, \ie, $\eta_{i+1} = \eta_i+\delta$.
In each of these steps, no new value arrives at the buffer, so $\forall j \in \{1, \dots, \delta-1\}\colon \neg \signal{valid\_its}^{\eta_i+j}$. 
Thus, by definition of $\abuffer$, the sum of the abstract entries always decreases by 1 for each \hclk cycle unless the buffer is already empty. 
In this case, the values does not change.
In cycle $\eta_{i+1}$, however, the buffer additionally receives a new value, so the sum of the entries also increases by $\dld(e_{i+1}) + 1$. 
Formally:
\begin{align*}
  \sum_{j=0}^\buffsize &\size(\abuffer^{\eta_{i+1}}) \\
  &= \sum_{j=1}^\buffsize \size(\abuffer^{\eta_{i+1}-1}) + (\dld(e_{i+1}) + 1) - 1 \\
  &= \sum_{j=1}^\buffsize \size(\abuffer^{\eta_{i}}) - \min\{\sum_{j=1}^\buffsize\size(\abuffer^{\eta_i}_j), \delta - 1\} + \dld(e_{i+1}) \\
  &= \backlog(e_i) - \min\{ \backlog(e_i), \delta - 1 \} + \dld(e_{i+1}) \tag{IH}\\
  &= \backlog(e_{i+1})
\end{align*}

The next fact can be proven using \Cref{lem:backlogabuff}:

\emph{2) The abstract buffer cannot overflow}, more concretely:
\begin{equation}\label{lem:abstractoverflow}
    \buffsize \geq \max \{\backlog\} \implies \forall \eta\colon \neg \signal{valid\_its}^\eta \lor \abuffer^{\eta-1}_1 = 0 \lor \abuffer^{\eta-1}_\buffsize = \bot
\end{equation}
Assume $\buffsize \geq \max\{\backlog\}$ and $\signal{valid\_its}^\eta \land \abuffer^{\eta-1}_1 > 0 \land \abuffer^{\eta-1}_\buffsize \neq \bot$. Since $\signal{valid\_its}^\eta$, we know that a new event arrived. 
If it is the first event, \ie, $\eta = \eta_0$, the contradiction follows from the definition of $\abuffer$. 
Otherwise, let $\eta = \eta_{i+1}$. 
We inspect the last $\delta$ steps. We know that no new value arrived, and because $\abuffer^{\eta-1}_\buffsize \neq \bot$ holds, there was no shift.
\begin{equation}\label{eq:deltadiff}
  \abuffer^{\eta_{i+1}-1-\delta}_1 = \abuffer^{\eta_i-1}_1 = \delta + \abuffer^{\eta_{i+1}-1}_1 \geq \delta+1 
\end{equation}
As a result:
\begin{align*}
  \backlog(e_i) 
    &= \sum_{j=1}^\buffsize\size(\abuffer^{\eta_i}) \tag{\cref{lem:backlogabuff}}\\
    &\geq \sum_{j=2}^\buffsize \size(\abuffer^{\eta_1}) + \delta + 1 \tag{\cref{eq:deltadiff}}\\
    &\geq \delta + 1 + \buffsize - 1 \tag{$\delta$>0}\\
    &\geq \buffsize+1 
\end{align*}
This contradicts $\buffsize \geq \max\{\backlog\}$.

\emph{3) Each entry of the abstract buffer corresponds to an entry in the concrete buffer.}
\begin{equation}\label{lem:abstractconcrete}
  \forall \eta,i\colon \abuffer^\eta_i = \bot \iff \buffer^\eta_i = \bot
\end{equation}
The equation holds by the definitions of $\abuffer$ and $\buffer$. The proof itself consists of correct bookkeeping of the buffer states and respective signal values.

By \Cref{lem:abstractconcrete} we know that each empty entry in the abstract buffer is also empty in the concrete buffer. Moreover, \Cref{lem:abstractoverflow} verifies that the abstract buffer never overflows. Thus, the concrete buffer cannot overflow as well, concluding the proof.

\end{proof}

%% file: figures/llc.tex
\begin{figure}[t]
  \centering
  \begin{tikzpicture}[scale=.78, every node/.style={transform shape}]

    \node [component] at (0,0)  (qif) {\component{LLQInterface}};
    \node [component] at (4.5,0) (ec) {\component{EvalController}};

    \path [signal] (4.5,1.5) -- (ec); 
    \node [nameright, text width=1.6cm] at (5.4,.8) {$\signal{d\_in} \in \bool[\sizeev]$\hfill};

    \path [signal] (.5,1.5) -- (.5,.5); 
    \node [nameright] at (2.2,.8) {$\signal{empty} \in \bool$\hfill};

    \path [signal] (-.5,.5) -- (-.5,1.5); 
    \node [nameleft, text width=1.2cm] at (-1.2,.8) {\hfill$\bool \ni \signal{pop}$};

    \path [signal, <->] (qif) -- (ec); 
    \node [nameright] at (3.4,.3) {$\signal{een} \in \bool$\hfill};

  \end{tikzpicture}
  \captionsetup{width=0.95\columnwidth} 
  \caption{Schematic of the Low-level Controller receiving event and deadline information from the queue and evaluating streams accordingly.}
  \label{fig:schematic:llc}
\end{figure}
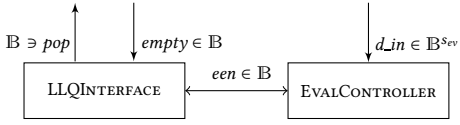

%% file: figures/evalcontrol.tex
\begin{figure}[t]
\vspace{-.3cm}
  \begin{minipage}{.48\columnwidth}
    \begin{subfigure}{.9\textwidth}
      \centering
      \begin{tikzpicture}[scale=.8, every node/.style={transform shape}]

          \node [state, initial above] (idle)                     {idle};
          \node [state] (pop)  [right=2cm of idle] {pop};
          \node [state] (eval) [below=2cm of pop]  {eval};

          \draw[thick, transition] (idle) edge              node [above]         {$\neg \signal{empty}$} (pop);
          \draw[thick, transition] (pop)  edge [bend right] node [right]         {$\top$}                (eval);
          \draw[thick, transition] (eval) edge [bend right] node [above, sloped] {$\neg\signal{empty}$}  (pop);
          \draw[thick, transition] (eval) edge              node [below, sloped] {$\signal{empty}$}      (idle);
        
      \end{tikzpicture}
      \caption{The \component{LLQInterface} handles the the communication with the queue in \smstate{pop}, and waits in \smstate{eval} until the evaluation finished.}
      \label{fig:statemachine:llqif}
    \end{subfigure}
  \end{minipage}
  \hfill
  \begin{minipage}{.48\columnwidth}
    \begin{subfigure}{.9\textwidth}
      \centering
      \begin{tikzpicture}[scale=.8, every node/.style={transform shape}]

          \node [state, initial above]              (idle)                             {idle};
          \node [state]              (one)     [right=2.5cm of idle]    {$1$};
          \node [state]              (twoone)  [below=2cm of one]       {$2.1$};
          \node [minimum height=1cm] (twodots) [left=.955cm of twoone]  {\dots};
          \node [state]              (twol)    [left=.955cm of twodots] {$2.\ell$};
        
          \path [thick, transition] (idle)    to node [above, sloped] {\signal{een}} (one);
          \draw [thick, transition] (one)     to node [above, sloped] {\signal{done\textsubscript{1}}} (twoone);
          \draw [thick, transition] (twoone)  to node [above, sloped] {\signal{done\textsubscript{2.1}}} (twodots);
          \draw [thick, transition] (twodots) to node [above, sloped] {\signal{done\textsubscript{2.($\ell$-1)}}} (twol);
          \draw [thick, transition] (twol)    to node [above, sloped] {\signal{done\textsubscript{2.$\ell$}}} (idle);
        
      \end{tikzpicture}
      \caption{The \component{EvalController} manages the evaluation. State~\smstate{1} treats input streams, \smstate{2.1} through \smstate{2.$\lambda$} output streams according to the evaluation order.}
      \label{fig:statemachine:ec}
    \end{subfigure}
  \end{minipage}
  \caption{State machines for the \component{LLQInterface} and the \component{EvalController}} 
  \label{fig:statemachines}
\end{figure}
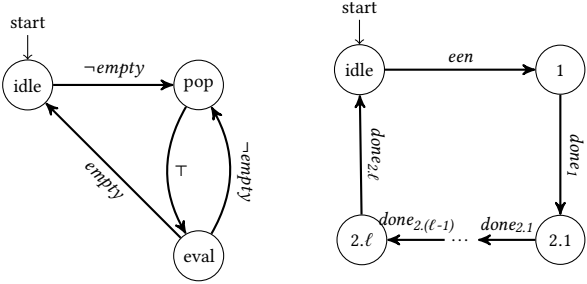

%% file: figures/streams.tex
\begin{figure}[t]
  \vspace{-.3cm}
  \centering
  \begin{tikzpicture}[scale=.8, every node/.style={transform shape}]

    \node [component] at (0,0)  (ini) {\component{in\textsubscript{i}}};
    \node [component] at (4,0)  (outj) {\component{out\textsubscript{j}}};
    \node [component] at (8,0)  (wetwa) {\component{w\textsubscript{$\eta$}}};

    \path [signal] (-.5,1.5) -- (-.5,.5); 
    \node [namecenter] at (-.6,1.6) {$\signal{upd} \in \bool$};
    \path [signal] (.5, 1.5) -- (.5, .5); 
    \node [namecenter] at (.2, 1.6) {$\signal{d\textsubscript{in}} \in \bool[s_i]$};

    \path [signal] (-.5,-.5) -- (-.5,-1.5); 
    \node [namecenter] at (-.2,-1.3) {$\signal{done} \in \bool$};
    \path [signal] (.5, -.5) -- (.5, -1.5); 
    \node [namecenter] at (.8,-1.3) {$\signal{d\textsubscript{out}} \in \bool[s_{\mathit{out}_i}]$};

    \path [signal] (3.1,1.5) -- (3.1, .5); 
    \node [namecenter] at (2.9,1.7) {$\signal{pe} \in \bool$};
    \path [signal] (3.7,1.5) -- (3.7, .5); 
    \node [namecenter] at (3.5,1.7) {$\signal{eval} \in \bool$};
    \path [signal] (4.3,1.5) -- (4.3, .5); 
    \node [namecenter] at (4.1,1.7) {$\signal{w\textsubscript{in}} \in \bool[\wdep(j)]$};
    \path [signal] (4.9,1.5) -- (4.9, .5); 
    \node [namecenter] at (4.7,1.7) {$\signal{dep\textsubscript{in}} \in \bool[\dep(j)]$};

    \path [signal] (3.5, -.5) -- (3.5,-1.5); 
    \node [namecenter] at (3.8,-1.3) {$\signal{done} \in \bool$};
    \path [signal] (4.5, -.5) -- (4.5,-1.5); 
    \node [namecenter] at (4.8,-1.3) {$\signal{d\textsubscript{out}} \in \bool[s_{\mathit{out}_j}]$};

    \path [signal] (7.1,1.5) -- (7.1, .5); 
    \node [namecenter] at (6.9,1.7) {$\signal{evict} \in \bool$};
    \path [signal] (7.7,1.5) -- (7.7, .5); 
    \node [namecenter] at (7.5,1.7) {$\signal{upd} \in \bool$};
    \path [signal] (8.3,1.5) -- (8.3, .5); 
    \node [namecenter] at (8.1,1.7) {$\signal{req} \in \bool$};
    \path [signal] (8.9,1.5) -- (8.9, .5); 
    \node [namecenter] at (8.7,1.7) {$\signal{d\textsubscript{in}} \in \bool[s_{\mathit{in}_\eta}]$};

    \path [signal] (7.5, -.5) -- (7.5,-1.5); 
    \node [namecenter] at (7.8,-1.3) {$\signal{done} \in \bool$};
    \path [signal] (8.5, -.5) -- (8.5,-1.5); 
    \node [namecenter] at (8.8,-1.3) {$\signal{d\textsubscript{out}} \in \bool[s_{\mathit{out}_\eta}]$};
    
  \end{tikzpicture}
  \caption{Input and output signals of input and output streams.}
  \label{fig:schematic:innercomps}
\end{figure}
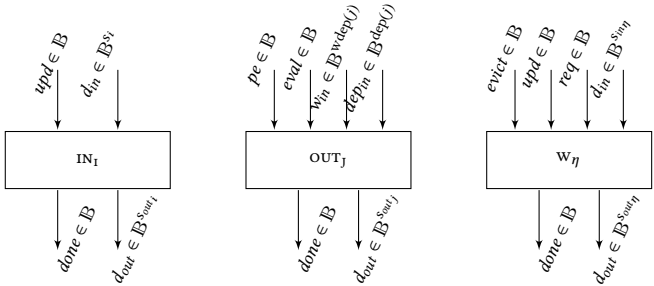

%% file: casestudy.tex
\section{Case Study}

We validated the compilation with three case studies. 
The first two monitor a network and an avionic and describe realistic scenarios, whereas the third one consists of synthetic data and emphasizes the benefits of the parallel evaluation structure presented in \Cref{sec:compilation}.
All specifications were compiled into VHDL code and then synthesized on a Zynq-Z-7010 ARM/FPGA SoC Trainer Board\footnote{\url{https://reference.digilentinc.com/reference/programmable-logic/zybo/reference-manual?_ga=2.102758273.1814454663.1555084001-1980681841.1546416239}}, which is logic-equivalent to an Artix-7 FPGA. 
The Zynq-7000 features 4.400 logic slices, each with four 6-bit input LUTs and 8 flip flops.

Note that the specifications in the benchmarks are simplified for illustration purposes. The current prototype does not support a floating or fixed point unit. The limitation is a result of technical incompatibilities in the Xilinx synthesizing software; from a theoretical standpoint, the inclusion of a floating-point unit is possible. This results however in a larger circuit realization of the specifications.

\subsection{Avionics}

\Cref{fig:spec:avionics} shows a specification for a drone. Input events consist of longitude and latitude values, the velocity and the number of GPS satellites in range. The GPS module is supposed to send values for the longitude and latitude with frequency \SI{10}{\hertz}. Output stream \lstinline{gps_freq} counts the number of samples received within a second and checks if it falls below 9. In this case, the first trigger reports the unexpectedly low sample frequency.
The second trigger reports a warning when the drone's velocity drops below 700, requiring that the velocity was greater than 700 before that. For the third trigger, we use a simplified reconstruction of the distance the drone traveled using the Pythagorean theorem. A more realistic approximation can be obtained \eg by using the haversine function. The square root computation is realized using the constant-time function proposed by Li and Chu~\cite{DBLP:conf/iccd/LiC96}. The distance is then discretely differentiated to compute the velocity according to the GPS module. This allows for cross-validating sensor values by comparing the sensed input velocity with the computed one. If the two values deviate too strongly, an alarm is raised. Lastly, we detect hover phases by integrating either velocity value and checking whether it lays below a threshold value.

We compiled the specification to VHDL and synthesized a circuit on the Zynq-7000 board. We report the resource consumption in terms of required flip-flops (FF), look-up tables (LUT), multiplexers (MUX), adders (CA), and multipliers (MULT) for each component below, where "Mon" describes the entire synthesized monitor:

\begin{tabular}{r | r r r r r}
  \toprule
  Component & FF & LUT & MUX & CA & MULT \\
  \midrule
  Mon & 3036 & 3685 & 26 & 656 & 18 \\
  HLC & 901  & 156  & 0  & 22  & 0 \\ 
  Q   & 543  & 442  & 0  & 43  & 0 \\
  LLC & 1281 & 2820 & 0  & 576 & 18 \\
  \bottomrule
\end{tabular}

Note that the amount of resources like flip-flops of the entire monitor is not equal to the sum of the resources of all components. The difference is required for internal tasks such as signal management.
One can see that most flip-flops reside in the \LLC because it manages the persisted values of all streams. The \HLC requires around 70\% as many, which can be contributed to the fact that each component of the \HLC contains internal registers while the greatest offset in the specification is only $-1$, reducing the memory requirement of the \LLC. The overwhelming majority of look-up tables, adders, and multipliers reside in the \LLC which was expected given that this component implements the evaluation logic. The 18 multipliers are required for squaring the $\delta$-values and computing the integral window.

The power consumption amounted to \SI{0.121}{\watt} when idle and \SI{1.620}{\watt} when processing.

We tested the monitor in online mode with sensor data created in a simulation using the ArduPilot\footnote{\url{http://ardupilot.org/}} Copter\footnote{\url{http://ardupilot.org/copter/index.html}} drone simulator. The simulator consisted of a multicopter flying over the campus of a university. Sensor information was piped to the monitor over a serial port. 
Evaluating events and periodic deadlines took on average 428 system clock cycles with a period $\clkperiod = \SI{100}{\mega\hertz}$.
Thus, each event took on average \SI{4.28}{\micro\second} to be processed.
Here, the worst slack amounted to \SI{1.653}{\nano\second}.

\subsection{Network Monitoring}\label{sec:casestudy:network}

The network monitoring exerted an immense pressure on the monitor due to the sheer amount of input data received in a short amount of time. In this setting it is also reasonable to forgo any assumption on the input frequency.

The specification in \Cref{fig:spec:network} fixes the IP of one particular server and checks network traffic based on the source and destination IP of requests, TCP flags, and the length of the payload. 
First, the \lstinline{length} stream is filtered based on whether the server is the target and the request pushes data. 
We sum up the filtered stream for a second and trigger an alert if the amount of data spikes over \SI{10}{\mega\byte}. 
Moreover, we count the number of opened and closed incoming connections and issue an alert if the server attempts to close more connections that were opened. 
Lastly, we check for a significant amount of incoming connections in a short amount of time. 

Due to the lower complexity of the specification, the resource consumption is also generally lower compared to the avionics example.
The number of look-up tables decreases by around 60\%, adders by 65\% and multipliers by 100\%. 
The number of flip-flops only decreases by around 38\% since there is no significant difference in the number of sliding windows and lookup expressions in the two specifications, but integral windows require 5-times as much memory as summation and count windows.

\begin{tabular}{r r r r r r}
  \toprule
  Component & FF & LUT & MUX & CA & MULT \\
  \midrule
  Mon & 1905 & 1533 & 23 & 226 & 0 \\
  HLC & 550  & 161  & 0  & 37  & 0 \\
  Q   & 330  & 342  & 0  & 28  & 0 \\
  LLC & 895  & 927  & 0  & 161 & 0 \\
  \bottomrule
\end{tabular}

The power consumption amounted to \SI{0.120}{\watt} when idle and \SI{1.570}{\watt} when processing, so there is no significant difference between the two specifications.

We tested the implementation with data from the Mid-Atlantic Collegiate Cyber Defense Competition (MACCDC)\footnote{\url{https://www.netresec.com/?page=MACCDC}}. We re-played the log data in real time using the time stamps provided.

While the evaluation process is simpler, the \HLC remains mostly the same. Thus, the amount of system clock cycles required per event only decreases by around 25\%, the response time for a single event is \SI{3.2}{\micro\second} on average. 
The worst slack time, however, increased by 150\% to \SI{4.0}{\nano\second}. 
This allows for safely increasing the system clock frequency by up to \SI{200}{\mega\hertz}. 
The reason for this is that the square root computation in the avionics specification has a significantly greater depth than all operation performed while monitoring the network. 
Since the computation is taken out in a single cycle, the slack time decreases significantly.

\subsection{Parallelization}
\Cref{sec:llc} presents a compilation that produces a highly parallel evaluation process by identifying modular structures within the specification.
The modularity is maximized when a specification contains a large number of independent streams.
Practical examples of this kind of specification are command-response or geofencing specifications. Here, each reaction and each face of the fence constitutes an independent stream, allowing for a parallel evaluation.

More concretely, consider a system that receives different commands from an external entity and needs to verify the system health depending on the kind of command.
Such a specification can be found in \Cref{fig:spec:parallel}. 
The highly disjunctive nature allows for perfect parallelization: each output stream solely depends on input streams.
In this case study, the specification is realized twice, once as proposed in \Cref{sec:compilation}, and once without the parallelization of the evaluation.
Purposefully declared spurious dependencies between successive output streams enforce a sequential evaluation.
\Cref{fig:spec:parallel} contains an extract of the specification.

Neither the size of the realization, nor the power consumption when idle varied between the realizations. A stress-test successively increases the input data rate until the \LLC can no longer process events in time. 
For this, the companion processor on the Zynq sends events to the FPGA and measures the time it takes for the FPGA to produce an output. 
This measurement produces more robust result than the communication over a bus in the preceding case studies but can only be applied in the absence of periodic streams.
When processing events in the maximum frequency for each realization, the parallel realization requires slightly more power (\SI{1.582}{\watt}) than the sequential one (\SI{1.581}{\watt}). 
As opposed to that, the execution time varies significantly. 
The sequential execution requires \SI{43.83}{\micro\second}, whereas the speed of the parallel execution exceeds the computation speed of the processor, which is up to \SI{866}{\mega\hertz}, \ie \SI{3.77}{\micro\second} between sending an event and attempting to read the output. 
As a result, the measured \SI{3.77}{\micro\second} constitute an upper bound on the actual response time. 
Practically, this means that if the processor sends events to the FPGA with it maximum frequency, the parallel realization can process all events, whereas the sequential one loses 89\% of the data.


\begin{figure}[t]
  \begin{lstlisting}[basicstyle=\footnotesize\ttfamily]
input lat, lon, velo: Int32
input gps: UInt8

output gps_freq@1Hz : bool := 
  lat.aggregate(over:1s,using:count).defaults(to:10) < 9
trigger gps_freq "GPS frequency less than 9 Hz"

output fast := velo > 700
trigger fast.offset(by:-1).defaults(to:false) & !fast
  "Slowing down"

output gps_dist := sqrt($\delta$(lon)^2 + $\delta$(lat)^2)
output gps_velo := gps_dist / $\delta$(time)
trigger abs(gps_velo - velo) > 10 "Sensor deviation"

output hovering@1Hz :=
  velo.aggregate(over:5s,using:$\int$).defaults(to:5) < 1
trigger hovering "Little distance covered"
  \end{lstlisting}
  \caption{\rtlola specification for monitoring a drone.}
  \label{fig:spec:avionics}
\end{figure}
\begin{figure}[t]
  \begin{lstlisting}[basicstyle=\footnotesize\ttfamily]
constant server: Int32 = ...
input src, dst: Int32
input fin, push, syn: bool
input length: Int32

output receiver := dst = server
trigger @1Hz
  receiver.aggregate(over:0.5s,using:$\Sigma$) > 10000
  "Many incoming connections"

output received := if receiver & push
  then 0
  else length
output workload@1Hz :=
  received.aggregate(over:1s,using:$\Sigma$)
trigger workload > 10^7 "Workload too high"

output opened := 
  open.offset(by:-1).defaults(to:0) + 
    (if dest = server & syn then 1 else 0)
output closed := 
  closed.offset(by:-1).defaults(to:0) + 
    (if dest = server & fin then 1 else 0)
trigger open - closed < 0
  "Closed more connection than were open"
  \end{lstlisting}
  \caption{\rtlola specification for monitoring network traffic.}
  \label{fig:spec:network}
\end{figure}
\begin{figure}[t]
  \begin{lstlisting}[basicstyle=\footnotesize\ttfamily]
input cmd: Int16
input height, x, y, ...: Int32
    
output health_crit_1: Bool := height < 400
trigger health_crit_1 $\land$ cmd = 1
    
    ...
    
output health_crit_512: Bool :=
  x > 700 $\lor$ y < 250 $\land$ height > 300
trigger health_crit_512 $\land$ cmd = 512
  \end{lstlisting}
  \caption{\rtlola specification for a highly parallelizable property.}
  \label{fig:spec:parallel}
\end{figure}

%% file: conclusion.tex
\section{Conclusion}\label{sec:conclusion}
We have presented a hardware-based monitoring approach for stream-based real-time specifications by compiling \rtlola specifications to circuits on FPGAs. The resulting circuits are small and efficient. Unlike interpreter-based approaches, the compiler limits the circuits to the operations in the specification and allows for a high degree of parallelization.
The presented case studies show that FPGA-based stream-monitoring is feasible for non-trivial specifications. While we used a small board, the available resources were only utilized by less than 50\% and the power consumption was around \SI{1.5}{\watt} under maximal pressure. This makes the approach suitable for integration into embedded systems without draining the available resources.

Building on the work presented in this paper, the next step is to extend the FPGA approach to stream specifications with parameterization~\cite{lola2} and to investigate the applicability of FPGA-based monitoring in distributed architectures. 